\documentclass[12pt, a4paper, reqno]{amsart}
\usepackage{mathrsfs}
\usepackage{bbm}
\usepackage{pgf,tikz}
\usepackage{tabu}
\usepackage{amsmath,amscd,amssymb,latexsym}
\usepackage[all]{xy}

\input xypic

\evensidemargin=\oddsidemargin
\DeclareMathVersion{can}
\DeclareMathAlphabet{\can}{OT1}{cmss}{m}{n}
\newtheorem{thm}{Theorem}[section]

\newtheorem{rem}[thm]{Remark}
\newtheorem{exa}[thm]{Example}
\theoremstyle{definition}

\theoremstyle{fact}

\theoremstyle{conjecture}

\numberwithin{equation}{section}

\newcommand{\Tr}{\operatorname{Tr}}

\begin{document}
\title[ Few-weight codes  over $\Bbb F_p+u\Bbb F_p$  and their distance optimal Gray image]
{Few-weight codes  over $\Bbb F_p+u\Bbb F_p$ associated with down sets and their distance optimal Gray image}


\author[Y. Wu]{ Yansheng Wu}

\address{\rm Department of Mathematics, Nanjing University of Aeronautics and Astronautics, Nanjing,
Jiangsu, 211100, P. R. China; Department of Mathematics, Ewha Womans University, 
52, Ewhayeodae-gil, Seodaemun-gu, Seoul, 03760, South Korea}
\email{wysasd@163.com}

\author[J. Y. Hyun]{Jong Yoon Hyun}
 \address{\rm Konkuk University, Glocal Campus, 268 Chungwon-daero Chungju-si Chungcheongbuk-do 27478, South Korea}
\email{hyun33@kku.ac.kr}





\date{\today}


\baselineskip=20pt

\begin{abstract}   

Let $p$ be an odd  prime number. In this paper, we construct $2(2p-3)$ classes of codes over the ring $R=\Bbb F_p+u\Bbb F_p,u^2=0$, which are associated with down sets.  We compute the Lee weight distributions of the $2(2p-3)$ classes of codes when the down sets are generated by a single maximal element. Moreover, by using the Gray map of the linear codes over $R$, we find out $2(p-1)$ classes of $p$-ary distance optimal linear codes. Two of them meet the Griesmer bound.



\end{abstract}

\maketitle

\bigskip
\section{Introduction}
     
Let $p$ be a prime number and $\Bbb F_p$ the finite field of order $p$. 
An $[n, k, d]$  linear code $\mathcal{C}$ of length $n$ over $\Bbb F_p$  is a $k$-dimensional subspace of $\Bbb F_p^n$ with minimum  Hamming  distance $d$. The dual $\mathcal{C}^{\perp}$ of $\mathcal{C}$ is defined by $\{x\in\mathbb{F}^n_p:x\cdot c=0 \text{ for all }c\in\mathcal{C}\}$, where $x\cdot c=x_1c_1+\cdots+x_nc_n\in\mathbb{F}_p$.
An $ [n,k,d]$ code  $\mathcal{C}$ is called distance optimal if no $[n,k,d+1]$ code exists, see \cite[Chapter 2]{HP}. It is well-known \cite{G} that $n\geq\sum_{i=0}^{k-1}\left\lceil\frac{d}{p^i}\right\rceil$, called the Griesmer bound, for any $[n,k,d]$ linear code over $\mathbb{F}_p$. It follows that a linear code over $\mathbb{F}_p$ satisfying the Griesmer bound is distance optimal. We say that a linear code is optimal if it attains the Griesmer bound.

Denote by $A_i$ the number of codewords in $\mathcal C$
 with Hamming weight $i$. The weight enumerator of  
 $\mathcal C$ is defined by
 $1+A_1z+A_2z^2+\cdots+A_nz^n.$
The sequence $(1, A_1, A_2, \ldots, A_n)$ is called the weight distribution of 
 $\mathcal C$. A code $\mathcal{C}$ is $t$-weight if the number of nonzero $A_{i}$ in the sequence $(A_1, A_2, \ldots, A_n)$ is equal to $t$.

Let $\Bbb F_{q}$ be the finite field of order $q$, where $q$ is a power of a prime $p$. Let $D=\{d_{1}, d_{2}, \ldots, d_{n}\}\subseteq \Bbb F_{w}$, where $w$ is a power of $q$. A linear code of length $n$ over $\Bbb F_{q}$ is defined by
\begin{equation*}
\mathcal{C}_{D}= \{(\Tr_{w/q}(xd_{1}), \ldots, \Tr_{w/q}(xd_{n})) : x\in  \Bbb F_{w}\},
\end{equation*}
where $\Tr_{w/q}$ is the trace function from $\Bbb F_{w}$ to $\Bbb F_{q}$.
  This generic construction was first introduced by Ding {\em et al.} \cite{D1, DN}. Many known few-weight linear codes could be produced by selecting a proper defining set $D$, see \cite{HY1, KY,LYL, LM2,Y}.   
 
 The construction method by Ding {\em et al}. can be generalized as follows: 
 Let $R$ be a finite commutative ring, $R_m$ be  an extension of $R$ of degree $m$ and $R_m^*$ be the multiplicative group of units of $R_m$. A trace code over $R$ with a defining set $L=\{l_1,l_2,\ldots, l_n\} \subseteq R_m^*$ is defined by
 \begin{equation*}\mathcal {C}_L=\{\Tr(xl_1), \Tr(xl_2), \ldots, \Tr(xl_n):x\in R_m\},\end{equation*}
where $\Tr(\cdot)$ is a linear function from $R_m$ to $R$. 

To derive few-weight linear codes, we choose specific commutative rings and their extensions. Now  
let $R=\Bbb F_q+u\Bbb F_q, u^2=0$, and $\mathcal {R}=\Bbb F_w+u\Bbb F_w$. The Lee weight distribution of a trace code $\mathcal{C}_L$  have been investigated in some literature.

 (1) When $R=\Bbb F_p+u\Bbb F_p, u^2=0$, $\mathcal {R}=\Bbb F_{p^m}+u\Bbb F_{p^m}$ and $L=\mathcal{Q}+u\Bbb F_{p^m}$, where $\mathcal{Q}$ is  the set of all square elements of $\Bbb F_{p^m}^*$,  the code $C_L$ is a two-weight or three-weight code, see \cite{S2}.
 
 (2) When $R=\Bbb F_p+u\Bbb F_p, u^2=u$, $\mathcal {R}=\Bbb F_{p^m}+u\Bbb F_{p^m}$ and $L=\mathcal{Q}+u\Bbb F_{p^m}^*$ or $\Bbb F_{p^m}^*+u\Bbb F_{p^m}^*$, where $\mathcal{Q}$ is  the set of all square elements of $\Bbb F_{p^m}^*$,  the code $C_L$ is a two-weight or few-weight code, see \cite{S3}.
 
 (3) When $R=\Bbb F_q+u\Bbb F_q, u^2=0$, $\mathcal {R}=\Bbb F_{q^m}+u\Bbb F_{q^m}$ and $L=C_0^{(e,q^m)}+u\Bbb F_{p^m}$, where $e$ is a divisor of $q-1$ and  $C_0^{(e,q^m)}$ is  the cyclotomic class of order $e$,  the code $C_L$ is a two-weight or few-weight code, see \cite{LM}.


In this paper, we study the following linear codes defined in $(1.1)$ right below.
Let $L$ be a subset of $\Bbb F_p^m+u\Bbb F_p^m,u^2=0$. A code $\mathcal {C}_{L}$ over $\Bbb F_p+u\Bbb F_p$ is defined by 
 \begin{equation}
 \mathcal {C}_{L}=\{c_{L}(\mathbf{a})=(\langle \mathbf{a}, \mathbf{l}\rangle )_{\mathbf{l}\in L}:\mathbf{a}\in \Bbb F_p^m+u \Bbb F_p^m\},
 \end{equation}
where $\langle \cdot, \cdot\rangle$ is the standard inner product on $ \Bbb F_p^m+u \Bbb F_p^m$.  Notice that if $\mathbf{x}=\mathbf{a}+u\mathbf{b}$ and $\mathbf{y}=\mathbf{c}+u\mathbf{d}$
 for $\mathbf{a},\mathbf{b},\mathbf{c},\mathbf{d}\in\mathbb{F}^m_p$, then 
 $\langle \mathbf{x},\mathbf{y} \rangle=\mathbf{a}\cdot\mathbf{c}+u(\mathbf{a}\cdot\mathbf{d}+\mathbf{b}\cdot\mathbf{c})$.

\begin{exa}
{\rm
Let $L=\{(1,0),(1,0)+(1,0)u,(1,0)+(0,1)u,(1,0)+(1,1)u\}$ be a subset of $\Bbb F_2^2+u\Bbb F_2^2,u^2=0$. Then
\[
\mathcal{C}_L=\{(a_1+b_1u,a_1+(a_1+b_1)u,a_1+(a_2+b_1)u,a_1+(a_1+a_2+b_1)u):
a_i,b_i\in\mathbb{F}_2,i=1,2\}.
\]
}
\end{exa}


One of the important problems in coding theory is to find the  $[n, k,d]$ linear codes over $\mathbb{F}_p$ having the largest minimum distance for given $n$ and $k$. In \cite{HKN}, the authors constructed some infinite families of distance optimal linear codes over $\mathbb{F}_p$ from down sets. 
The aim of this paper is to construct the few-weight codes $\mathcal{C}_L$ over $R=\mathbb{F}_p+u\mathbb{F}_p, u^2=0$ and find out the distance optimal linear codes over $\mathbb{F}_p$ from the Gray image of $\mathcal{C}_L$, where $L$'s are subsets of $R^m$ associated with  down sets of $\mathbb{F}^m_p$ generated by a single maximal element.


The rest of this paper is organized as follows. Section 2  recalls basic concepts and introduce some known results. In Sections 3, we determine the Lee weight distribution of $2(2p-3)$ classes of codes (Theorems 3.1-3.4). In Section 4, by using the Gray map, we obtain $2(p-1)$ classes of distance optimal linear codes (Theorems 4.1, 4.2), from which we obtain Table 5 of distance optimal linear codes with small dimensions. Two of them meet the Griesmer bound. In Section 5, we conclude the paper.

\section{Preliminaries}

Firstly, the Lee weight defined on $R^m=(\Bbb F_p+u\Bbb F_p)^m, u^2=0$
and the Gray map from $R^m$ to $\mathbb{F}^{2m}_p$ are introduced.
Next, we define a down set of $\mathbb{F}^m_p$ by endowing a partial order on $\mathbb{F}^m_p$.

In the remainder of this paper,  we always assume that  $R=\Bbb F_p+u\Bbb F_p$, where $u^2=0$.

\subsection{Lee weight and Gray map}$~$
 
 By a code of length $m$ over $R$, we mean a subset of $R^m$. 
  A linear code of $\mathcal{C} $ of length $m$ over $R$ is an $R$-submodule of $R^m$. The inner product between $\mathbf{x}=(x_1,x_2,\ldots, x_m)$ and $\mathbf{y}=(y_1,y_2, \ldots, y_m)\in R^m$ is defined by $\langle \mathbf{x},\mathbf{y} \rangle=\sum_{i=1}^mx_iy_i\in R$. Notice that if $\mathbf{x}=\mathbf{a}+u\mathbf{b}$ and $\mathbf{y}=\mathbf{c}+u\mathbf{d}$
 for $\mathbf{a},\mathbf{b},\mathbf{c},\mathbf{d}\in\mathbb{F}^m_p$, then 
 $\langle \mathbf{x},\mathbf{y} \rangle=\mathbf{a}\cdot\mathbf{c}+u(\mathbf{a}\cdot\mathbf{d}+\mathbf{b}\cdot\mathbf{c})$.
 
The Gray map $\hat{\phi}$ from $R$ to $\Bbb F_p^2$ is defined by 
$$\hat{\phi}: R\to \Bbb F_p^2, a+ub\mapsto (b,a+b),~a,b\in \Bbb F_p$$
This leads to the Gray map $\phi$ naturally from $R^m$ to $\Bbb F_p^{2m}$ as follows:
 $$\phi: R^m\to \Bbb F_p^{2m},~\mathbf{x}=\mathbf{a}+u\mathbf{b}\mapsto (\mathbf{b}, \mathbf{a}+\mathbf{b}).$$
 
The Hamming weight of a vector $\mathbf{a}$ of length $m$ over $\Bbb F_p$ is defined to be the number of nonzero entries in the vector $\mathbf{a}$. The Lee weight of a vector $\mathbf{x}=\mathbf{a}+\mathbf{b}u$ of length $m$ over $R$ is defined to be the Hamming weight of its Gray image as follows:
$$w_L(\mathbf{x})=w_L(\mathbf{a}+u\mathbf{b})=w_H(\mathbf{b})+w_H(\mathbf{a}+\mathbf{b}).$$
The Lee distance $d_L(\mathbf{x},\mathbf{y})$ of between two vectors $\mathbf{x,y}\in R^m$ is defined as $w_L(\mathbf{x-y})$. It is easy to check that the Gray map $\phi$ is an isometry from $(R^m, d_L)$ to $(\Bbb F_p^{2m}, d_H)$, where $d_H$ denotes the Hamming distance. Obviously, if $\mathcal{C}$ is a $\mathbb{F}_p$-submodule of $R^m$ with parameters $(n,p^k,d)$, then $\phi(\mathcal{C})$ is a linear code over $\mathbb{F}_p$ with parameters $[2n,k,d]$.

\begin{exa}
{\rm
We continue Example 1.1 to illustrate the Gray image of $\mathcal{C}_L$.
\begin{align*}
 \phi(\mathcal{C}_L)=\{(b_1,a_1+b_1,a_1+b_1,b_1,a_2+b_1,a_1+a_2+b_1,\\
 a_1+a_2+b_1,a_2+b_2):a_i,b_i\in\mathbb{F}_2,i=1,2\}.   
\end{align*}
We point out that the minimum distance of $\phi(\mathcal{C}_L)^{\perp}$ is two.
}
\end{exa}
\subsection{Down sets}$~$


 Let $v=(v_1, \ldots, v_m)$ and $w=(w_1, \ldots, w_m)$ be two vectors in $\Bbb F_p^m$. We define a partial order on $\Bbb F_p^m$ as follows:  $v\preceq w$ if and only if $v_i\le w_i$ for all $i\in [m]=\{1,\ldots, m\}$. We say that a subset $\Delta$ of $\Bbb F_p^m$ is a down set if $w\in \Delta$ and $v\preceq w$ imply $v\in \Delta$.  Then $(\Bbb F_p^m, \preceq)$ forms a complete lattice, where the join and the meet of two  vectors $v$ and $w$ in $\Bbb F_p^m$ are  respectively defined by $v \lor w=(\max\{v_1, w_1\}, \ldots,\max\{v_m, w_m\} )$
 and $v\land w=(\min\{v_1, w_1\}, \ldots,\min\{v_m, w_m\} )$. 
An element $v\in \Delta$ is maximal if $v\preceq w$ and $w\in \Delta $ imply $v=w$. It is readily verified that every down sets $\Delta $ of $\Bbb F_p^m$ is generated by the set of maximal elements of $\Delta$, i.e., $\Delta=\langle v(1), \ldots, v(t)\rangle$, where $\{v(1), \ldots, v(t)\}$ is the set of maximal element of $\Delta$.


\section{The Lee weight distributions }
Hereafter, we assume that $p$ is an odd prime number.
Let $\Delta$ be a down set of $\Bbb F_p^m$ and $ L=\Delta^c+u\Bbb F_p^m$, where $\Delta^c=\Bbb F_p^m\backslash \Delta$.  
Recall from $(1.1)$ that
 \begin{equation*}
 \mathcal {C}_{L}=\{c_{L}(\mathbf{y})=(\langle \mathbf{y}, \mathbf{l}\rangle )_{\mathbf{l}\in L}:\mathbf{y}\in \Bbb F_p^m+u \Bbb F_p^m\}\\
 =\{(\mathbf{a}\cdot\mathbf{c}+u(\mathbf{a}\cdot\mathbf{d}+\mathbf{b}\cdot\mathbf{c}))_{\mathbf{c}\in\Delta^c,\mathbf{d}\in\mathbb{F}^m_p}:\mathbf{a},
 \mathbf{b}\in\mathbb{F}^m_p\}.
 \end{equation*}
  The length of the code $\mathcal{C}_{L}$ is $|L|$. 
  Notice that $\mathcal{C}_L$ is not linear but $\mathbb{F}_p$-submodule of $R^m$.
 The Gray image $\phi(\mathcal{C}_L)$ of $\mathcal{C}_L$ is
  \[
  \phi(\mathcal{C}_L)=\{(\mathbf{a}\cdot\mathbf{d}+\mathbf{b}\cdot\mathbf{c},\mathbf{a}\cdot\mathbf{d}+\mathbf{b}\cdot\mathbf{c}+\mathbf{a}\cdot\mathbf{c})_{\mathbf{c}\in\Delta^c,\mathbf{d}\in\mathbb{F}^m_p}:\mathbf{a},
 \mathbf{b}\in\mathbb{F}^m_p\}.
  \]








Assume that $\mathbf{a}={}\alpha+u{\beta}$, $\mathbf{l}_1={t_1}+u{y}$, and $\mathbf{l}_2={t_2}+u{y}$, where  ${\alpha}=(\alpha_1, \ldots, \alpha_m),$ ${\beta}=(\beta_1, \ldots, \beta_m),$ ${y}=(y_1, \ldots, y_m)$ $\in \Bbb F_p^m$, ${t_1}\in \Delta$, and ${t_2}\in \Delta^c$ without expressing in the bold face. If $\mathbf{a}=\mathbf{0}$, then $w_L(c_{L}(\mathbf{a}))=w_L(c_{L^c}(\mathbf{a}))=0$. Next we  assume that $\mathbf{a}\neq \mathbf{0}$. Then the Lee weight of the codeword $c_{L^c}(\mathbf{a})$ of $\mathcal{C}_{L^c}$ becomes that 
\begin{eqnarray}
&&w_L(c_{L^c}(\mathbf{a}))
=w_L(({\alpha}\cdot{t_1}+u({\alpha}\cdot{y}+{\beta}\cdot{t_1}))_{{t_1}\in \Delta, {y}\in \Bbb F_p^m})\nonumber\\
&=&w_H(({\alpha}\cdot{y}+{\alpha}\cdot{t_1})_{{t_1}\in \Delta, {y}\in \Bbb F_p^m})+w_H((({\alpha}+{\beta})\cdot{t_1}+\alpha\cdot{y})_{{t_1}\in \Delta, {y}\in \Bbb F_p^m})\nonumber\\
&=&2|L^c|-\frac1p\sum_{x\in\Bbb F_p}\sum_{t_1\in \Delta}\sum_{{y}\in \Bbb F_p^m}\zeta_p^{(\alpha\cdot{y}+\beta\cdot{t_1})x}-\frac1p\sum_{x\in\Bbb F_p}\sum_{t_1\in \Delta}\sum_{{y}\in \Bbb F_p^m}\zeta_p^{((\alpha+\beta)\cdot{t_1}+\alpha\cdot{y})x}\nonumber\\
&=&2|L^c|(1-\frac 1p)-\frac1p\sum_{x\in \Bbb F_p^*}\sum_{t_1\in \Delta}\zeta_p^{\beta\cdot{t_1}x}\sum_{{y}\in \Bbb F_p^m}\zeta_p^{\alpha\cdot{xy}}
-\frac1p \sum_{x\in \Bbb F_p^*}\sum_{t_1\in \Delta}\zeta_p^{(\alpha+\beta)\cdot{t_1}x}\sum_{{y}\in \Bbb F_p^m}\zeta_p^{\alpha\cdot{xy}}\nonumber\\
&=&2|L^c|(1-\frac 1p)-p^{m-1}\delta_{0,\alpha}(\sum_{x\in \Bbb F_p^*}\sum_{t_1\in \Delta}\zeta_p^{\beta\cdot{t_1}x}
+ \sum_{x\in \Bbb F_p^*}\sum_{t_1\in \Delta}\zeta_p^{(\alpha+\beta)\cdot{t_1}x}),
\end{eqnarray}
where $\delta$ is the Kronecker delta function.

Similarly, the Lee weight of the codeword $c_{L}(\mathbf{a})$ of $\mathcal{C}_{L}$ becomes that \begin{eqnarray}
&&w_L(c_{L}(\mathbf{a}))
=w_H((\alpha\cdot{y}+\beta\cdot{t_2})_{{t_2}\in \Delta^c, {y}\in \Bbb F_p^m})+w_H(((\alpha+\beta)\cdot{t_2}+\alpha\cdot{y})_{{t_2}\in \Delta^c, {y}\in \Bbb F_p^m})\nonumber\\
&=&2|L|(1-\frac 1p)-\frac1p\sum_{x\in \Bbb F_p^*}\sum_{t_2\in \Delta^c}\zeta_p^{\beta\cdot{t_2}x}\sum_{{y}\in \Bbb F_p^m}\zeta_p^{\alpha\cdot{xy}}
-\frac1p \sum_{x\in \Bbb F_p^*}\sum_{t_2\in \Delta}\zeta_p^{(\alpha+\beta)\cdot{t_2}x}\sum_{{y}\in \Bbb F_p^m}\zeta_p^{\alpha\cdot{xy}}\nonumber\\
&=&2|L|(1-\frac 1p)-p^{m-1}\delta_{0,\alpha}(\sum_{x\in \Bbb F_p^*}\sum_{t_2\in \Delta^c}\zeta_p^{\beta\cdot{t_2}x}
+ \sum_{x\in \Bbb F_p^*}\sum_{t_2\in \Delta^c}\zeta_p^{(\alpha+\beta)\cdot{t_2}x}).
\end{eqnarray}

Note that $|L^c|+|L|=p^{2m}$. Then \begin{equation}w_L(c_{L^c}(\mathbf{a}))+w_L(c_{L}(\mathbf{a}))=2p^{2m-1}(p-1)-p^{2m-1}(p-1)\delta_{0,\alpha}(\delta_{0,\beta}+\delta_{0,\alpha+\beta}).\end{equation}

By using $(3.1)$, $(3.2)$ and $(3.3)$, we give the Lee weight distribution of the code  $\mathcal{C}_{L}$  in the case that the down set  is generated by a single maximal element.

We begin with down sets of the simplest forms to determine their Lee weight distributions. 

\begin{thm} Let $m\ge 2$ be an integer and $p$ be an odd prime number. Let $\Delta=\langle (r, 0, \ldots, 0) \rangle$ be a down set of $\Bbb F_p$ for $r=1, \ldots, p-1$.  Then the code  $\mathcal{C}_{L}$ has length $p^{m}(p^m-r-1)$, size $p^{2m}$, and its Lee weight distribution is given by Table 1. 
\end{thm}

\begin{table}[h]
\caption{Lee weight distribution of the code $\mathcal{C}_{L}$  in Theorem 3.1}   
\begin{tabu} to 0.6\textwidth{X[1.2,c]|X[1,c]}  
\hline 
\rm{Lee Weight}&\rm{Frequency}\\ 
\hline
$0$&$1$\\ 
\hline
$2p^{2m-1}(p-1)$& $p^{m-1}-1$ \\ 
\hline
$2p^{m}(p^m-p^{m-1}-r)$& $p^{m-1}(p-1)$ \\ 
\hline
$2p^{m-1}(p-1)(p^m-r-1)$& $p^m(p^m-1)$ \\ 
\hline
\end{tabu}  
\end{table}

\begin{proof}
It is easy to check that the length  of the code $\mathcal{C}_{L}$ is  $|L|=p^{m}(p^m-r-1)$. Let $\mathbf{a}=\alpha+u\beta$ for $\alpha=(\alpha_1,\ldots,\alpha_m)$ and $\beta=(\beta_1,\ldots,\beta_m)\in\mathbb{F}^m_p$.
If $\alpha\neq 0$, then   $c_{L}(\mathbf{a})=2|L|(1-\frac 1p)=2p^{m-1}(p-1)(p^m-r-1)$.
If $\alpha=0$, then by Eq. (3.1), we have \begin{eqnarray*}c_{L^c}(\mathbf{a})&=&2p^{m-1}(p-1)(r+1)-2p^{m-1}\sum_{x\in \Bbb F_p^*}\sum_{t_1\in \Delta}\zeta_p^{\beta\cdot{t_1}x}\\
&=&2p^{m-1}(p-1)(r+1)-2p^{m-1}\sum_{x\in \Bbb F_p^*}(1+\zeta^{\beta_1x}+\cdots+\zeta^{\beta_1xr})\\
&=&\left\{
   \begin{array}{ll}
   0, &\mbox{ if $ \beta_1=0$},\\
   2p^mr, &\mbox{ if $\beta_1\neq 0$}.\\
    \end{array}\right.
\end{eqnarray*}
By Eq. (3.3), we have \begin{eqnarray*}c_{L}(\mathbf{a})&=& 2p^{2m-1}(p-1)-c_{L^c}(\mathbf{a})\\
&=&\left\{
   \begin{array}{ll}
   2p^{2m-1}(p-1), &\mbox{ if $ \beta_1=0$},\\
   2p^m(p^m-p^{m-1}-r), &\mbox{ if $\beta_1\neq 0$}.\\
    \end{array}\right.
\end{eqnarray*}
The frequency of each codeword of the codes should be computed by the vector $\mathbf{a}$. \end{proof}



\begin{thm} Let $m\ge 3$ be an integer and $p$ be an odd prime number. Let $\Delta=\langle (p-1, r, 0,\ldots, 0) \rangle$ be a down set of $\Bbb F_p$ for $r=1, \ldots, p-1$.  Then the code  $\mathcal{C}_{L}$ has length $p^{m}(p^m-p(r+1))$, size $p^{2m}$, and its Lee weight distribution is given by Table 2.



\end{thm}


\begin{table}[h]
\caption{Lee weight distribution of the code $\mathcal{C}_{L}$  in Theorem 3.2}   
\begin{tabu} to 0.7\textwidth{X[2,c]|X[1,c]}  
\hline 
\rm{Lee Weight}&\rm{Frequency}\\ 
\hline
$0$&$1$\\ 
\hline
$2p^{2m-1}(p-1)$&$p^{m-2}-1$\\
\hline
$2p^{m+1}(p^{m-1}-p^{m-2}-r)$& $p^{m-2}(p-1)$ \\ 
\hline
$2p^{m}(p-1)(p^{m-1}-r-1)$& $p^{2m}-p^{m-1}$ \\ 
\hline
\end{tabu}  
\end{table}

\begin{proof} It is easy to check that the length  of the code $\mathcal{C}_{L}$ is  $|L|=p^{m}(p^m-p(r+1))$. Let $\mathbf{a}=\alpha+u\beta$ for $\alpha=(\alpha_1,\ldots,\alpha_m)$ and $\beta=(\beta_1,\ldots,\beta_m)\in\mathbb{F}^m_p$. If $\alpha\neq 0$, then  $c_{L}(\mathbf{a})=2|L|(1-\frac 1p)=2p^{m-1}(p-1)(p^m-p(r+1))$.
If $\alpha=0$, then by Eq. (3.1), we have \begin{eqnarray*}c_{L^c}(\mathbf{a})&=&2p^{m}(p-1)(r+1)-2p^{m-1}\sum_{x\in \Bbb F_p^*}\sum_{t_1\in \Delta}\zeta_p^{\beta\cdot{t_1}x}\\
&=&2p^{m}(p-1)(r+1)-2p^{m-1}\sum_{x\in \Bbb F_p^*}\sum_{y\in \Bbb F_p}\zeta_p^{\beta_1xy}\sum_{z=0}^{r}\zeta_p^{\beta_2xz}\\
&=&\left\{
   \begin{array}{llll}
   0, &\mbox{ if $\beta_1=\beta_2= 0$},\\
   2p^{m+1}r, &\mbox{ if $\beta_1=0$}  \mbox{ and } \mbox{ $\beta_2\neq0$},\\
   2p^{m}(p-1)(r+1), &\mbox{ if $\beta_1\neq 0$} .
    \end{array}\right.
\end{eqnarray*} 
By Eq. (3.3), we have 
\begin{eqnarray*}c_{L}(\mathbf{a})&=& 2p^{2m-1}(p-1)-c_{L_1}(\mathbf{a})\\
&=&\left\{
   \begin{array}{llll}
   2p^{2m-1}(p-1), &\mbox{ if $\beta_1=\beta_2= 0$},\\
   2p^{m}(p^m-p^{m-1}-r), &\mbox{ if $\beta_1=0$}  \mbox{ and } \mbox{ $\beta_2\neq0$},\\
   2p^{m}(p-1)(p^{m-1}-r-1), &\mbox{ if $\beta_1\neq 0$} .
    \end{array}\right.
\end{eqnarray*}
The frequency of each codeword of the codes should be computed by the vector $\mathbf{a}$. 
\end{proof}

\begin{thm} Let $m\ge 3$ be an integer and $p$ be an odd prime number. Let $\Delta=\langle (p-2, r, 0,\ldots, 0) \rangle$ be a down set of $\Bbb F_p$ for $r=1, \ldots, p-2$.  Then the code  $\mathcal{C}_{L}$ has length $p^{m}(p^m-(p-1)(r+1))$, size $p^{2m}$, and its Lee weight distribution is given by  Table 3.



\end{thm}


\begin{table}[h]
\caption{Lee weight distribution of the code $\mathcal{C}_{L}$  in Theorem 3.3}   
\begin{tabu} to 0.8\textwidth{X[2,c]|X[1.4,c]}  
\hline 
\rm{Lee Weight}&\rm{Frequency}\\ 
\hline
$0$&$1$\\ 
\hline
$2p^{2m-1}(p-1)$&$p^{m-2}-1$\\ 
\hline
$2p^m(p-1)(p^{m-1}-r)$& $p^{m-2}(p-1)$ \\ 
\hline
$2p^{2m-1}(p-1)-2p^{m}(p-2)(r+1)$& $p^{m-2}(p-1)(p-r)$ \\ 
\hline
$2p^{2m-1}(p-1)-2p^{m}(pr+p-2r-1)$& $p^{m-2}r(p-1)$ \\ 
\hline
$2p^{m-1}(p-1)(p^m-(p-1)(r+1))$& $p^{m}(p^m-1)$ \\ 
\hline
\end{tabu}  
\end{table}

\begin{proof} It is easy to check that the length of the code $\mathcal{C}_{L}$ is  $|L|=p^{m}(p^m-(p-1)(r+1))$. Let $\mathbf{a}=\alpha+u\beta$ for $\alpha=(\alpha_1,\ldots,\alpha_m)$ and $\beta=(\beta_1,\ldots,\beta_m)\in\mathbb{F}^m_p$.
 If $\alpha\neq 0$, then $c_{L}(\mathbf{a})=2|L|(1-\frac 1p)=2p^{m-1}(p-1)(p^m-(p-1)(r+1))$. 
 Note that $\sum_{x\in \Bbb F_p^*}\sum_{t_1\in \Delta}\zeta_p^{\beta{t_1}x}$ has been determined in \cite[Theorem 4.11]{HKN}. Table 3 follows from Eq. (3.3). 
\end{proof}

\begin{thm} Let $m\ge 3$ be an integer and $p$ be an odd prime number. Let $\Delta=\langle (p-3, r, 0,\ldots, 0) \rangle$ be a down set of $\Bbb F_p$ for $r=1, \ldots, p-2$.  Then the code  $\mathcal{C}_{L}$ has length $p^{m}(p^m-(p-2)(r+1))$, size $p^{2m}$, and its Lee weight distribution is given by Table 4.



\end{thm}


\begin{table}[h]
\caption{Lee weight distribution of the code $\mathcal{C}_{L}$  in Theorem 3.4}   
\begin{tabu} to 0.9\textwidth{X[2,c]|X[1.4,c]}  
\hline 
\rm{Lee Weight}&\rm{Frequency}\\ 
\hline
$0$&$1$\\ 
\hline
$2p^{2m-1}(p-1)$&$p^{m-2}-1$\\ 
\hline
$2p^{2m-1}(p-1)-2p^m(p-2)r$& $p^{m-2}(p-1)$ \\ 
\hline
$2p^{2m-1}(p-1)-2p^{m}(p-3)(r+1)$& $p^{m-2}(p-1)(\frac{p+1}{2}-r+\lfloor \frac r 2\rfloor)$ \\ 
\hline
$2p^{2m-1}(p-1)-2p^{m}(pr+p-3r-2)$& $p^{m-2}(p-1)(p-1-2\lfloor \frac  r 2\rfloor)$ \\ 
\hline
$2p^{2m-1}(p-1)-2p^{m}(pr+p-3r-1)$& $p^{m-2}(p-1)(\lfloor \frac  r 2\rfloor+r-\frac{p-1}{2})$ \\ 
\hline
$2p^{m-1}(p-1)(p^m-(p-2)(r+1))$& $p^{m}(p^m-1)$ \\ 
\hline
\end{tabu}  
\end{table}

\begin{proof} It is easy to check that the length  of the code $\mathcal{C}_{L}$ is   $|L|=p^{m}(p^m-(p-2)(r+1))$. Let $\mathbf{a}=\alpha+u\beta$ for $\alpha=(\alpha_1,\ldots,\alpha_m)$ and $\beta=(\beta_1,\ldots,\beta_m)\in\mathbb{F}^m_p$.
If $\alpha\neq 0$, then  $c_{L}(\mathbf{a})=2|L|(1-\frac 1p)=2p^{m-1}(p-1)(p^m-(p-2)(r+1))$.
Note that $\sum_{x\in \Bbb F_p^*}\sum_{t_1\in \Delta}\zeta_p^{\beta{t_1}x}$ has been determined in \cite[Theorem 4.14]{HKN}.  Table 4 follows from Eq. (3.3). 
\end{proof}


\section{Optimal codes and examples}

Recall that the Gray map $\phi$  is an isometry from $(R^m, d_L)$ to 
$(\Bbb F_p^{2m}, d_H)$.



\begin{thm}
Let $m\ge 3$ be an integer and $p$ be an odd prime number. 
Let $L=\Delta^c+u\mathbb{F}^m_p,u^2=0$ for $\Delta=\langle(r,0,\ldots,0)\rangle$, $r=1,2,\ldots,p-1$.
Then the Gray image $\phi(\mathcal{C}_{L})$ of $\mathcal{C}_{L}$ is a distance optimal linear code and the minimum distance of $\phi(\mathcal{C}_{L})^{\perp}$ is two.
In particular, if $p<2r+2$ and $2(r+1)(p-1)<p^2$, i.e., $r=
\frac{p-1}{2}$, then  $\phi(\mathcal{C}_{L})$ meets the Griesmer bound. 

\end{thm}

\begin{proof} 
By Theorem 3.1, the code $\phi(\mathcal{C}_{L})$ has the following parameters:
$$n=2p^m(p^m-r-1),~ k=2m,~ d=2p^{m-1}(p-1)(p^m-r-1).$$ 
Firstly, we show that $\phi(\mathcal{C}_{L})$ is distance optimal. 
Assume to the contrary that there is an $[n,k,d+1]$ code. We have that
\begin{eqnarray*}
&&\sum_{i=0}^{2m-1}\bigg\lceil {\frac{2p^{m-1}(p-1)(p^m-r-1)+1}{p^i}} \bigg\rceil\nonumber\\
&=&\sum_{i=0}^{m-1}\bigg\lceil {\frac{2p^{m-1}(p-1)(p^m-r-1)+1}{p^i}} \bigg\rceil\\
&+&\sum_{i=m}^{2m-1}\bigg\lceil {\frac{2p^{2m-1}(p-1)}{p^i}
+\frac{1-2p^{m-1}(p-1)(r+1)}{p^i}} \bigg\rceil\nonumber\\
&=&2(p^m-1)(p^m-r-1)+m+\sum_{i=0}^{m-1}2p^i(p-1)+\sum_{i=1}^{m}\bigg\lceil{\frac{1}{p^{m+i}}-\frac{2(r+1)(p-1)}{p^i}} \bigg\rceil.
\end{eqnarray*}
Since $1\le r\le p-1$, we have $1\le \bigg\lceil {\frac{2(r+1)}{p}} \bigg\rceil\le 2$ and  
$0<\frac{2(r+1)(p-1)}{p^i} <2,~(i=2,3,\ldots,m),$
so that 
\[
\bigg\lceil{\frac{1}{p^{m}}-\frac{2(r+1)(p-1)}{p}} \bigg\rceil
=-2(r+1)+\bigg\lceil{\frac{1}{p^{m}}+\frac{2(r+1)}{p}} \bigg\rceil
=-2(r+1)+2=-2r
\]
and
\[
\bigg\lceil{\frac{1}{p^{m+i}}-\frac{2(r+1)(p-1)}{p^i}} \bigg\rceil\geq-1~(i=2,3,\ldots,m).
\]
It follows that
\[
\sum_{i=0}^{m-1}\bigg\lceil{\frac{1}{p^{m+i}}-\frac{2(r+1)(p-1)}{p^i}} \bigg\rceil\geq -2r-(m-1),
\]
and so
\begin{multline*}\sum_{i=0}^{2m-1}\bigg\lceil {\frac{2p^{m-1}(p-1)(p^m-r-1)+1}{p^i}} \bigg\rceil \ge 2(p^m-1)(p^m-r-1)+m+2(p^m-1)\\-2r-(m-1)=2p^m(p^m-r-1)+1,
\end{multline*}
which contracts to the Griesmer bound. 

Secondly, we show that the minimum distance $d^{\perp}$ of $\phi(\mathcal{C}_L)^{\perp}$ is two. Assume that $d^{\perp}\geq3$. By the sphere packing bound and $r\leq p-1$, we have 
\[
p^{2m}\geq|\phi(\mathcal{C}_L)^{\perp}|(1+2p^m(p^m-r-1)(p-1))
>2p^m(p^m-(p-1)-1)(p-1),
\]
equivalently, $(p^{m-1}-1)(2p-1)<1,$
which is a contradiction. We now claim that there is no codeword in $\phi(\mathcal{C}_L)^{\perp}$ whose Hamming weight is one. Assume to the contrary that there is a codeword in $\phi(\mathcal{C}_L)^{\perp}$ whose Hamming weight is one. Then the $i$th coordinate position of any codeword in $\phi(\mathcal{C}_L)$ is the zero for some $i$, and so for fixed $\mathbf{c}\in\Delta^c$ and $d\in\mathbb{F}^m_p$, we have either $\mathbf{a}\cdot\mathbf{d}+\mathbf{b}\cdot\mathbf{c}=0$ or $\mathbf{a}\cdot\mathbf{d}+\mathbf{b}\cdot\mathbf{c}
+\mathbf{a}\cdot\mathbf{c}=0$ for all $\mathbf{a},\mathbf{b}\in\mathbb{F}^m_p$. 
In any case, we derive that $\mathbf{c}=\mathbf{0}$, which is a contradiction with $\mathbf{c}\in\Delta^c$.

It remains to prove the third part.
We see that 
\begin{eqnarray*}
&&\sum_{i=0}^{2m-1}\bigg\lceil {\frac{2p^{m-1}(p-1)(p^m-r-1)}{p^i}} \bigg\rceil\nonumber\\
&=&\sum_{i=0}^{m-1}\bigg\lceil {\frac{2p^{m-1}(p-1)(p^m-r-1)}{p^i}} \bigg\rceil+\sum_{i=m}^{2m-1}\bigg\lceil {\frac{2p^{2m-1}(p-1)}{p^i}-\frac{2p^{m-1}(p-1)(r+1)}{p^i}} \bigg\rceil\nonumber\\
&=&2(p^m-1)(p^m-r-1)+\sum_{i=0}^{m-1}2p^i(p-1)-\sum_{i=1}^m\bigg\lfloor{\frac{2(r+1)(p-1)}{p^i}} \bigg\rfloor.
\end{eqnarray*}
 Since $p<2r+2$ and $2(r+1)(p-1)<p^2$, we have that 
 $$ \bigg\lfloor {\frac{2(r+1)(p-1)}{p}} \bigg\rfloor=
 \bigg\lfloor {2(r+1)-\frac{2(r+1)}{p}} \bigg\rfloor
 =2(r+1)-\bigg\lceil\frac{2(r+1)}{p}\bigg\rceil=2(r+1)-2$$ and 
 $$\bigg\lfloor {\frac{2(r+1)(p-1)}{p^i}}\bigg\rfloor=0~ (i=2,,3,\ldots,m).$$
It follows that 
\[
\sum_{i=1}^m\bigg\lfloor{\frac{2(r+1)(p-1)}{p^i}} \bigg\rfloor
=(-2(r+1)+2)+0=-2r,
\]
and so
\begin{multline*}
\sum_{i=0}^{2m-1}\bigg\lceil {\frac{2p^{m-1}(p-1)(p^m-r-1)}{p^i}}
\bigg\rceil=2(p^m-1)(p^m-r-1)+2(p^m-1)-2r\\
=2p^m(p^m-r-1),
\end{multline*}
which shows that if $p<2r+2$ and $2(r+1)(p-1)<p^2$, then
$\phi(\mathcal{C}_{L})$ meets the Griesmer bound. This completes the proof.
\end{proof}

We have the following theorem in a similar computation of Theorem $4.1$ with
$n=2p^{m}(p^m-p(r+1)),~k=2m,~d=2p^m(p-1)(p^m-r-1)$.

\begin{thm}
Let $m\ge 3$ be an integer and $p$ be an odd prime number. 
Let $L=\Delta^c+u\mathbb{F}^m_p,u^2=0$ for $\Delta=\langle(p-1,r,0,\ldots,0)\rangle$, $r=1,2,\ldots,p-1$.
Then the Gray image $\phi(\mathcal{C}_{L})$ of $\mathcal{C}_{L}$ is a distance optimal linear code and the minimum distance of $\phi(\mathcal{C}_{L})^{\perp}$ is two.
In particular, if $p<2r+2$ and $2(r+1)(p-1)<p^2$, i.e., $r=
\frac{p-1}{2}$, then  $\phi(\mathcal{C}_{L})$ meets the Griesmer bound. 
\end{thm}

\begin{rem}
{\rm
We point out that in the database of Grassl \cite{G1}, he provides a complete list of distance optimal $[n,k]$ codes with small lengths when  $p\in\{3,5,7\}$, and we see that most of the distance optimal linear codes are unknown when $n\geq31$ and $k\geq 8$, where the upper bound of $n$
is restrictive and relies on $p$. We have constructed $2(p-1)$ classes of distance optimal linear codes in Theorems 4.1 and 4.2. It is believed that our distance optimal linear codes include new codes although we can not compare with the parameters in Grassl's table \cite{G1} because our code length is large. 

Table 5 is obtained by using Theorem 4.1 and 4.2, and $*$ stands for a linear code attaining the Griesmer bound.
}
\end{rem}

\begin{center}
\begin{table}[h]
\centering \caption{(Distance) optimal linear codes $\phi(\mathcal{C}_L)$ in Theorems 4.1 and 4.2}


{\scriptsize
\[ \begin{array}{|c||c|c|c|c||c|c|c|c|}  \hline
p &  n & k & d & \mbox{Optimality} & n & k & d & \mbox{Optimality}\\
\hline
3 &  1350 & 6 &900& \mbox{Optimal*}  & 12798 & 8 & 8532 & \mbox{Optimal*}   \\
3 & 1296  & 6 & 864 & \mbox{Distance optimal} & 12636  & 8 & 8428 &\mbox{Distance optimal}\\
3 & 1134  & 6 & 756   &\mbox{Optimal*} & 12150 & 8 & 8100&\mbox{Optimal*}   \\
3 & 972   & 6 &  648   &\mbox{Distance optimal} & 11664 & 8 & 7776 &\mbox{Distance optimal}\\
5 & 30750  & 6 & 24600    &\mbox{Distance-optimal}& 778750 & 8 &623000& \mbox{Distance optimal}  \\
5 &  30500 & 6 & 24400  & \mbox{Optimal*}  & 777500 & 8  &622000& \mbox{Optimal*} \\
5  & 30250 & 6 &24200&\mbox{Distance optimal}   & 776250& 8   &621000&\mbox{Distance optimal} \\
5  & 30000 & 6 &24000 & \mbox{Distance optimal}&775000 & 8    &620000&\mbox{Distance optimal} \\
5 & 28750  & 6 & 23000& \mbox{Distance optimal}   &768750 & 8     &615000& \mbox{Distance optimal}\\
5  &27500  & 6 &22000&  \mbox{Optimal*}   & 762500& 8    &610000& \mbox{Optimal*}\\
5 &26250 & 6 &21000&\mbox{Distance optimal} & 756250& 8  &605000& \mbox{Distance optimal} \\
5  & 25000 & 6 &20000& \mbox{Distance optimal}  &750000 & 8 &600000&  \mbox{Distance optimal} \\
7  & 233926 & 6 &200508&\mbox{Distance optimal}   &11519998 & 8  &9874284&  \mbox{Distance optimal}\\
7 & 233240 & 6 &199920&\mbox{Distance optimal}  &11515196 & 8  &9870168&\mbox{Distance optimal} \\
7  & 232554 & 6 &199332& \mbox{Optimal*} & 11510394 & 8 &9866052&  \mbox{Optimal*}  \\
7  &231868  & 6 &198744& \mbox{Distance optimal}  &11505592 & 8 &9861936& \mbox{Distance optimal}\\
7&231182  & 6 &198156& \mbox{Distance optimal}  &11500790 & 8  &9857820&\mbox{Distance optimal} \\
7  &230496 & 6 &197568&\mbox{Distance optimal}   & 11495988& 8 &9853704& \mbox{Distance optimal} \\
7 & 225694 & 6 &193452&\mbox{Distance optimal}    &11462374 & 8  &9824892 & \mbox{Distance optimal} \\
7 & 220892 & 6 &189336&\mbox{Distance optimal}     &11428760 & 8  &9796080& \mbox{Distance optimal}  \\
7 &216090  & 6 &185220 & \mbox{Optimal*}    &11595146 & 8 &9767268&\mbox{Optimal*} \\
7 &211288  & 6 &181104&\mbox{Distance optimal}     & 11361532& 8  &9738456& \mbox{Distance optimal}  \\
7 &206486  & 6 &176988 &\mbox{Distance optimal}     &11327918 & 8 &9709644&\mbox{Distanceoptimal}  \\
7 &210684  & 6 &172872&\mbox{Distance optimal}     & 11294304& 8  &9680832&\mbox{Distance optimal}   \\
\hline
\end{array} \]
}
\end{table}
\end{center}

The following are two numeral examples.

\begin{exa}
{\rm
Let $m=3, p=3$, and $r=1$ in Theorem 3.3. 
Then  $\phi(\mathcal{C}_{L})$ is a four-weight ternary code of  parameters $[1242, 6, 810]$ with the weight enumerator $$1+6z^{810}+702z^{828}+18z^{864}+2z^{972}.$$
}
\end{exa}

\begin{exa} 
{\rm
Let $m=3, p=5$, and $r=2$ in Theorem 3.4. 
Then   $\phi(\mathcal{C}_{L})$ is a five-weight pentary code of  parameters $[29000, 6, 23000]$ with the weight enumerator $$1+20z^{23000}+15500z^{23200}+40z^{23250}+60z^{23500}+4z^{25000}.$$
}
\end{exa}

\section{Concluding remarks}
The main contributions of this paper are the following

\begin{itemize}
\item Construction of the linear codes $\mathcal{C}_L$ over $\Bbb F_p+u\Bbb F_p$, where $u^2=0$ and $p$ is an odd prime number, which is defined in Eq. (1.1) associated with down sets;

\item  Determination of the Lee weight distributions of the codes $\mathcal{C}_L$ over $\Bbb F_p+u\Bbb F_p$ in the case that down sets are all generated by a single maximal element (Theorems 3.1, 3.2, 3.3 and 3.4);

\item Some infinite families of  $p$-ary  optimal linear   codes  from the Gray image of the codes $\mathcal{C}_L$ over $\Bbb F_p+u\Bbb F_p$ (Theorems 4.1 and 4.2).
\end{itemize}


Finally, we would inform the reader that in Section 3 we just determined the Lee weight distributions of the codes  $\mathcal{C}_L$ when the down sets are   generated by a single maximal element.
The reader is cordially invited to consider  the general cases.


\bigskip
\section*{Acknowledgments}

The second author was supported by the National Research Foundation of
Korea(NRF) grant funded by the Korea government(MEST)
(NRF-2017R1D1A1B05030707). We express our gratitude to the reviewers for their very helpful comments, which improved the
exposition of this paper.

\end{document}